\documentclass[12pt]{article}

\usepackage[utf8]{inputenc}
\usepackage[T1]{fontenc}
\usepackage[standard]{ntheorem}

\usepackage[english]{babel}
\usepackage{amsmath}
\usepackage{amssymb}
\usepackage{graphicx}
\usepackage{hyperref}

\usepackage[english,rules]{ALgo}
\usepackage{url}

\usepackage{fp}
\usepackage{tikz}
\usetikzlibrary{arrows,automata,shapes,snakes}

\newcommand{\drawColoredWord}[5]{
\FPeval{x}{(#1)}
\FPeval{y}{(#2)}
\FPeval{wordLength}{(#3)}
\FPeval{z}{(\x+\wordLength)}
\FPeval{t}{(\y+.5)}
\FPeval{u}{(((\x)+(\z))/2)}
\FPeval{v}{(\y+.25)}
\draw[fill=#5!30] (\x,\y) rectangle (\z,\t);
\draw (\u,\v) node {#4};
}

\def\sa#1{\mbox{\tt #1}}

\def\sh1{\textit{sh1}}

\def\pp{\mathinner{\ldotp\ldotp}}	% deux points

\def\suff{\mathit{suff}}
\def\lsc{\mathit{lcs}}
\def\conds{\textit{Cond-suf}}
\def\condo{\textit{Cond-occ}}
\def\bsuff{\textit{good-suff}}
\def\suf{\textit{suff}}
\def\alph{\textit{alph}}

\begin{document}

\title{A fast implementation of the good-suffix array for the Boyer-Moore string matching algorithm}

\author{Thierry Lecroq\footnote{CURIB,
UFR des Sciences et des Techniques,
Université de Rouen Normandie,
76821 Mont-Saint-Aignan Cedex,
France,
\protect\url{Thierry.Lecroq@univ-rouen.fr},
\protect\url{monge.univ-mlv.fr/~lecroq/lec_en.html}
}
}

\date{Univ Rouen Normandie, INSA Rouen Normandie, Université Le Havre Normandie, Normandie Univ, LITIS UR 4108, CNRS NormaSTIC FR 3638, IRIB, F-76000 Rouen, France}

\maketitle

\begin{abstract}
String matching is the problem of finding all the occurrences of a
 pattern in a text.
It has been intensively studied and the Boyer-Moore string matching
 algorithm is probably one of the most famous solution to this problem.
This algorithm uses two precomputed shift tables called the good-suffix
 table and the bad-character table.
The good-suffix table is tricky to compute in linear time.
Text book solutions perform redundant operations.
Here we present a fast implementation for this good-suffix table
 based on a tight analysis of the pattern.
 Experimental results show two versions of this new implementation are the
  fastest in almost all tested situations.
\end{abstract}

%***********************************************************************
%***********************************************************************
\section{Introduction}

The string matching problem consists of finding one or more usually
 all the occurrences of a pattern $x=x[0\pp m-1]$ of length $m$ in a text $y=y[0\pp n-1]$
 of length $n$ where both strings $x$ and $y$ are build on the same alphabet.
It can occur in many applications, for instance in information retrieval, bibliographic search and
 molecular biology.
It has been extensively studied and numerous techniques and algorithms
 have been designed to solve this problem
% (see~\cite{CR94,Ste94,Gus97,CHL2001,Smy2002,NR2002,CL97b,CHL2007}).
 (see~\cite{CL2004,Faro2011a}).
We are interested here in the problem where the pattern is given first
 and can then be searched in various texts.
This specific problem is called exact online string matching.
Thus a preprocessing phase is allowed on the pattern.

Basically an online string matching algorithm uses a window to scan the text.
The size of this window is equal to the length of the pattern.
It first aligns the left ends of the window and the text.
Then it checks if the pattern occurs in the window (this specific work
 is called an {\em attempt}) and {\em shifts} the window to the right.
It repeats the same procedure again until the right end of the window
 goes beyond the right end of the text.
The many different solutions differ in the way they compare the content
 of the window and the pattern, in the way they compute the lengths of the shifts and
 in the information that are stored from one attempt to the next.

There is no universal algorithm in the sense that the efficiency of a particular algorithm
 mainly depends on the size of the alphabet and the length of the pattern~\cite{FL2010}.
According to~\cite{Faro2011a}, the many algorithms can be classified in character comparison
 based algorithms,
automata based algorithms and bit-parallelism based algorithms.

The Boyer-Moore string matching
 algorithm~\cite{BM77}, which is a character comparison based algorithm,  is probably one of the most famous solution to this problem.
It has inspired a lot of subsequent works (\cite{Gal79,Hor80,AG86,HS91,Lec92} just to cite a few).
This algorithm uses two precomputed shift tables called the good-suffix
 array and the bad-character array.
The good-suffix table is tricky to compute in linear time.
Actually the original paper by Boyer and Moore did not specify how to compute it.
In~\cite{KMP77}, the proposed algorithm is not correct in all cases.
The first correct algorithm was given by Rytter~\cite{Ryt80}.
Text book solutions perform redundant operations.
It uses another table, called $\suff$ that gives for every position of the pattern
 the longest suffix of the pattern ending at that position.
The classical method for computing the good-suffix table consists
 of first computing the $\suff$ table and to perform two successive scans
 of this $\suff$ array to compute the good-suffix table that involves three
  loops on the interval $[0;m[$.

Actually it is easy to see that the computation can be done with only two loops.
Redundant operations can be spotted and thus avoided.
But how fast is it?
Experimental results show that it is fastest in almost all
 situations.
It is surprising that more than 50 years after the publication of the Boyer-Moore algorithm
 we can still design new algorithm for computing this good-suffix table
  and it was a lot of fun doing so!
 
The rest of this article is organized as follows.
Section~\ref{sect-pre} presents basic notions and notations on strings.
Section~\ref{sect-classic} presents the classical computation of the good-suffix table.
Section~\ref{sect-ft} give three new methods for computing this table.
In Section~\ref{sect-expe} experimental results are presented and
 we give our conclusions in Section~\ref{sect-conc}.

%%%%%%%%%%%%%%%%%%%%%%%%%%%%%%%%%%%%%%%%%%%%%%%%%%%%%%%%%%%%%%%%%%%%%%%%%%%%%%%
%%%%%%%%%%%%%%%%%%%%%%%%%%%%%%%%%%%%%%%%%%%%%%%%%%%%%%%%%%%%%%%%%%%%%%%%%%%%%%%
\section{\label{sect-pre}Preliminaries}

An {\em alphabet} is a finite set of elements called {\em letters} or
 {\em symbols}.
A {\em string} is a sequence of zero or more symbols from an
 alphabet $\Sigma$ of size $\sigma$; the string with zero symbols is denoted by
 $\varepsilon$.
The set of all strings over the alphabet $\Sigma$ is
 denoted by $\Sigma^*$.
%We consider an alphabet of size $s$;
% for $1 \leq i \leq s$, $\sigma[i]$ denotes the $i$-th symbol of
% $\Sigma$.
A string $x$ of length $n$ is represented
 by $x[0\pp n-1]$, where $x[i]\in\Sigma$ for $0\le i\le n-1$.
The subset of $\Sigma$ of letters occurring in a string $x$ is
 denoted by $\alph(x)$.
A string $u$ is a {\em prefix} of $x$ if $x=uw$ for $w\in\Sigma^*$.
Similarly, $u$ is a {\em suffix} of $x$ if $x=wu$ for 
 $w\in\Sigma^*$.
A string $u$ is a {\em border} of $x$ if $u$ is both a prefix and a
 suffix of $x$ and $u \neq x$.
 An integer $p$ such that $1\le p \le |x|$ is a {\em period} of a string $x$
 if $x[i] = x[i+p]$ for every position $i$ such that $0\|e i < |x|-p$.
Borders and periods are dual notions since if $x=uv$ and $u$ is
 a border of $x$ then $|v|$ is a period of $x$.
Furthermore during the run of a exact online string matching algorithm
 the length of the optimal shift to apply when an occurrence of a pattern is found
 is given by the shortest period (usually called {\em the border}) of the pattern.

%%%%%%%%%%%%%%%%%%%%%%%%%%%%%%%%%%%%%%%%%%%%%%%%%%%%%%%%%%%%%%%%%%%%%%%%%%%%%%%
%%%%%%%%%%%%%%%%%%%%%%%%%%%%%%%%%%%%%%%%%%%%%%%%%%%%%%%%%%%%%%%%%%%%%%%%%%%%%%%
\section{\label{sect-classic}Classical computation}

When an attempt $T$ takes place at right
 position $j$ on the text $y$, the window
 contains the factor $y[j-m+1 \pp j]$ of the text $y$. 
The index $j$ is thus the right position of the factor of the text contained in the window.
The longest common suffix of two strings $u$ and $v$ being denoted by
 $\lsc(u,v)$
 for an attempt $T$ at position $j$ on the text $y$, we set
 $z = \lsc(y[0 \pp j],x)$
  and $d$ the length of the
  shift
  applied just after the attempt $T$. 

The general situation at the end of the attempt $T$ is the following:
 the suffix $z$ of $x$ has been identified in the text $y$ and, if $|z|<|x|$,
 a negative comparison occurred between the letter
 $a=x[m-|z|-1]$ of the pattern and the letter $b=y[j-|z|]$ of the text.
In other words, by setting $i=m-|z|-1$,
 we have $z=x[i+1\pp m-1]=y[j-m+i+2\pp j]$ and, 
 either $i=-1$,
 or $i\geq0$ with $a=x[i]$, $b=y[j-m+i+1]$ and $a\ne b$
 (see \figurename{}~\ref{figu-situ-bm}).

Taking into account the information collected on the text $y$ during
 the attempt, the natural shift to apply consists in aligning the
 factor $z$ of the text with its rightmost occurrence in $x$ preceding by a letter different from $a$. 
If there is no such occurrence, the shift should take into account the longest prefix of $x$ that is
 also a suffix of $z$. 
%These two cases are illustrated in Figure~\ref{bm-ms}.

In the two situations that have just been examined, the computation of the
 shift following $T$ is independent of the text.
It can be previously computed for each position of the pattern.
To this aim, we define two conditions that correspond to the case where
 the string $z$ is the suffix $x[i+1 \pp m-1]$ of $x$.
They are the \textit{suffix condition}
 and the \textit{occurrence condition}.
They are defined, for every position $i$ on $x$, every shift
 $d$ of $x$ and every letter $b\in\Sigma$, by:
\[
 \conds(i,d)=\begin{cases}
  0<d \leq i+1 \mbox{ and } x[i-d+1 \pp m-d-1] \mbox{ is a suffix of } x \cr
  \mbox{or} \cr
  i+1<d \mbox{ and } x[0 \pp m-d-1] \mbox{ is a suffix of } x
 \end{cases}
\]
 and
\[
 \condo(i,d)=\begin{cases}
  0<d \leq i \mbox{ and } x[i-d]\neq x[i] \cr
  \mbox{or} \cr
  i<d.
 \end{cases}
\]
Then, the
 \textit{good-suffix} table, denoted by
 $\bsuff$, is defined in the following
 way: for every position $i$ on $x$,
$$
\bsuff[i] = \min \{ d \mid \conds(i,d) \mbox{ and } \condo(i,d)
 \mbox{ hold}\}
$$
for $0\leq i \leq m-1$.

The computation of the good-suffix table is not straightforward.
One method consists in using a table called $\suf$ defined as follows:
$$\suf[i] = |\lsc(x, x[0\pp i])|$$
for $0\leq i \leq m-1$.

The relation between the table $\bsuff$ and the table $\suf$ is the following

$$\bsuff[i] = \min\{m-1-j \mid m-1 - \suf[j] = i\}$$

for $0\leq i \leq m-1$ (see \figurename{}~\ref{figu-rel-tables}) when $z=x[i+1\pp m-1]$
 reoccurs in $x$
 otherwise
$$\bsuff[i] \leq \min\{m-1-j \mid \suf[j] = j+1 \mbox{ and } \suf[j] < m-1-i\}$$
for $0\leq i \leq m-1$ (see \figurename{}~\ref{figu-rel-tables2}) when $z=x[i+1\pp m-1]$
 does not reoccur in $x$.
The first condition $\suf[j] = j+1$ means that $x[0\pp j]$ is both a prefix of $x$ and a suffix of $x$
 thus a border and then $m-1-j$ is a period of $x$.
The second condition $\suf[j] < m-1-i$ means that the border is shorter than $z=x[i+1\pp m-1]$
 (the shift cannot take into account information about suffixes of $x$ longer than $z$).
It can be noticed that $\bsuff[0]$ is equal to the period of $x$ and can be used to shift the window
 when an occurrence of the pattern has been found.

Table $\suf$ can be computed by the algorithm
 \Algo{Suffixes} given in \figurename{}~\ref{figu-algo-suffixes}
  and table $\bsuff$ can be computed by the algorithm
 \Algo{Suffixes} given in \figurename{}~\ref{figu-algo-good-suffixes}.

Both algorithms \Algo{Suffixes} and \Algo{Good-Suffixes} runs in linear time
 (see~\cite{CHL2007} for the detailed correctness and complexity analysis).
Algorithm \Algo{Suffixes} computes the values of the table $\suf$ from right to left.
When computing $\suf[i]$, for $0\le i \le m-2$, all values
 $\suf[j]$ have already been computed for $i < j <m$.
It keeps two variables $f$ and $g$ such that
 $i<f$ and $g=\suf[f]-f$ is minimal.
Thus if $g<i$ and if $\suf[m-1-f+i]\ne i-g$
 then $\suf[i] = \min\{\suf[m-1-f+i], i-g\}$ without
 having to compare letters.
Only when $\suf[m-1-f+i]= i-g$ comparing letters is necessary
 to compute $\suf[i]$.
This mainly contributes to the linear time complexity of Algorithm \Algo{Suffixes}.
The linearity of the computation of table $\bsuff$ given table $\suf$
 is straightforward.

An example of computation of the $\bsuff$ table
% for $x = \sa{aabbaaaabbaaaaabbaaabbaaaa}$
 is given in \figurename{}~\ref{figu-example1}.

%%%%%%%%%%%%%%%%%%%%%%%%%%%%%%%%%%%%%%%%%%%%%%%%%%%%%%%%%%%%%%%%%%%%%%%%%%%%%%%
%%%%%%%%%%%%%%%%%%%%%%%%%%%%%%%%%%%%%%%%%%%%%%%%%%%%%%%%%%%%%%%%%%%%%%%%%%%%%%%
\section{\label{sect-ft}Improved computation}

Algorithm \Algo{Good-Suffixes} performs two scans of the table $\suf$ in order
 to compute the table $\bsuff$ after that one scan was necessary to compute table $\suf$.
We will show that both tables can be computed in the same time with only two scans of the interval
 $[0;m[$ where $m$ is the length of the pattern.
Let us first remark that for any string $x$ of length $m$, $\suf[m-1]$ is equal to $m$ and this value
 is never used in the computation of $\bsuff$ thus we can avoid to compute it.
From now on, let us denote by $a\in\Sigma$ the last letter of $x$.
Let us then observe that for every $i$ on $x$ such that $0\le i \le m-2$, 
 if $x[i]\ne a$ then $\suf[i]=0$ and only the rightmost such position
 is used in the computation of $\bsuff$.
Let us assume that $x$ contains at least two distinct symbols ($|\alph(x)| \ge 2$).

Then the proposed improved computation of the table $\bsuff$ will consists in
 computing the tables $\suf$ and $\bsuff$ together with a single scan of $x$
 from right to left and dealing only with occurrences of its last symbol
 (positions $i$ where $x[i]=a$ for $0\le i \le m-2$).

Furthermore, if $x$ ends with a run of $a$ of length $k_1$,
 runs of $a$ of length different from $k_1$ in $x$ will be
 easily handle.

The next lemma shows how to handle the positions of the last run of $a$
 (i.e. the run that ends $x$, see \figurename{}~\ref{figu-lemma1}).

\begin{lemma}\label{lemma-lastrun}
Let $x=x_1b_1a^{k_1}$ for $x_1\in \Sigma^*$, $a,b_1\in \Sigma$, $b_1\ne a$ and an integer $k_1 \ge 1$.
Let $\ell_1=|x_1b_1|$ and $r_1=|x_1b_1a^{k_1}|-2=m-2$, then
\begin{enumerate}
\item
$\suf[i] = i-\ell_1+1$;
\item
$\bsuff[i] = i-\ell_1+1$;
\item
$\bsuff[m-1] = k_1$;
\end{enumerate}
for $\ell_1 \le i \le r_1$.
\end{lemma}

\begin{proof}
For $\ell_1 \le i \le r_1$, $i-\ell_1+1 < k_1=r_1-\ell_1+2$.

\begin{enumerate}
\item
Let $k=i-\ell_1+1 < k_1$. 
Then $x[\ell_1\pp i] = x[m-k\pp m-1] = a^k$ is a suffix of $x$ and $x[\ell_1-1]=b_1\ne a = x[m-1-k]$,
 thus $\suf[i] = k$;
\item
$\conds(i,k)$ holds since $0<k\le i+1$ and $x[i-k+1\pp m-k-1] = a^{k_1-k}$ is a suffix of $x$.
$\condo(i,k)$ holds since $0<k\le i$ and $x[i-k] = b_1 \ne a = x[i]$.
For any $0<d<k$, $\condo(i,d)$ does not hold since $x[i-d] = a = x[i]$
 thus $\bsuff[i] = k$;
\item
For any $0<d<k_1$, $\condo(m-1,d)$ does not hold since $x[m-1-d] = a = x[m-1]$
 and 
$\conds(m-1,k_1)$ trivially holds since $0<k_1\le m$ and $x[m-1-k_1+1\pp m-k_1-1] = \varepsilon$ is the empty suffix of $x$.
 thus $\bsuff[m-1] = k_1$.
\end{enumerate}
\end{proof}

The next two lemmas show how to handle the case of subsequent runs of $a$.
We start by runs of length $k_2$ strictly less than $k_1$ (see \figurename{}~\ref{figu-lemma2}).

\begin{lemma}\label{lemma-smallrun}
If $x=x_2b_2a^{k_2}x_3b_1a^{k_1}$ for $x_2,x_3\in \Sigma^*$, $a,b_1,b_2\in \Sigma$, $b_1\ne a$, $b_2\ne a$ and two integers $k_1,k_2 \ge 1$ such that
 $k_2 < k_1$.
Let $\ell_2=|x_2b_2|$ and $r_2=|x_2b_2a^{k_2}|-1$, then
\begin{enumerate}
\item
$\suf[i] = i-\ell_2+1$;
\item
$\bsuff[m-1-\suf[i]] < m-1-i$;
\end{enumerate}
for $\ell_1 \le i \le r_1$.
\end{lemma}

\begin{proof}
For $\ell_2 \le i \le r_2$, let $k=i-\ell_2+1$.

\begin{enumerate}
\item
Then $x[\ell_2\pp i] = x[m-k\pp m-1] = a^k$ is a suffix of $x$ and $x[\ell_2-1]=b_2\ne a = x[m-1-k]$,
 thus $\suf[i] = k$;
\item
$\conds(m-1-k,k_1-k)$ holds since $0<k_1-k\le m-k$ and $x[m-1-k-(k_1-k)+1\pp m-(k_1-k)-1]=x[m-k_1\pp m-k_1+k-1] = a^k$ is a suffix of $x$.
$\condo(m-1-k,k_1-k)$ holds since $0<k_1-k\le m-1-k$ and $x[m-1-k-(k_1-k)] = x[m-k_1-1] = b_1 \ne a = x[m-1-k]$
 and $k _1-k< m-1-i$ thus $\bsuff[m-1-k] \le m-1-i$;
\end{enumerate}
\end{proof}

The next lemma shows how to handle runs of length greater or equal than $k_1$
  (see \figurename{}~\ref{figu-lemma3}).

\begin{lemma}\label{lemma-largerun}
If $x=x_4b_4a^{k_3}x_5b_1a^{k_1}$ for $x_4,x_5\in \Sigma^*$, $a,b_1,b_4\in \Sigma$, $b_1\ne a$, $b_4\ne a$ and two integers $k_1,k_3 \ge 1$ such that
 $k_3 > k_1$.
Let $\ell_3=|x_4b_4|$ and $r_3=|x_4b_4a^{k_3}|-1$, and let $e = \ell_3+k_1-1$ then:
\begin{enumerate}
\item
$\suf[i] = k_1$ for $f<i\le r_3$;
\item
$\suf[f] \ge k_1$;
\item
$\suf[i] = i-\ell_3+1$ for $\ell_3\le i < e$;
\item
if $k_3>k_1$ then $\bsuff[m-1-\suf[r_3]] \le m-1-r_3$;
\item
$\bsuff[m-1-\suf[e]] \le m-1-e$;
\item
$\bsuff[m-1-\suf[i]] < m-1-i$ for $e<i<r_3$ and for $\ell_3 \le i < f$
\end{enumerate}
\end{lemma}

\begin{proof}
\begin{enumerate}
\item
%$\suf[i] = k_1$ for $f<i\le r_3$;
For $f<i\le r_3$, $x[i-k_1+1\pp i]=a^{k_1}=x[m-k_1 \pp m-1]$ is a suffix of $x$ and $x[i-k_1]=a\ne b = x[m-1-k_1]$
 thus $\suf[i] = k_1$;
\item
%$\suf[f] \ge k_1$;
$x[f-k_1+1\pp f]=a^{k_1}=x[m-k_1\pp m-1]$ is a suffix of $x$
 thus $\suf[f] \ge k_1$;
\item
%$\suf[i] = i-\ell_3+1$ for $\ell_3\le i < f$;
For $\ell_3 \le i < e$, $x[\ell_3\pp i]=a^{i-\ell_3+1}=x[m-i-\ell_3 \pp m-1]$ is a suffix of $x$ and $x[\ell_3-1]=b_4\ne a = x[m-i-\ell_3-1]$
 thus $\suf[i] = i-\ell_3+1$;
\item
%$\bsuff[m-1-\suf[r_3]] \le m-1-r_3$;

if $k_3>k_1$ then $\suf[r_3] = k_1<k_3$.

$\conds(m-1-k_1,m-1-r_3)$ holds since $x[m-1-k_1-(m-1-r_3)+1 \pp m-1-(m-1-r_3)]=x[r_3-k_1+1 \pp r_3]=a^{k_1}$ is a suffix of $x$
 and $\condo((m-1-k_1,m-1-r_3)$ holds since $0< m-1-r_3<m-1-k_1$ and $x[m-1-k_1-(m-1-r_3)]=x[r_3-k_1]=a\ne b_1=x[m-1-k_1]$.
 
\item
%$\bsuff[m-1-\suf[f]] \le m-1-f$;

Let $\suf[s]=\lsc(x[0\pp e],x) = s$.
Then either $s = f+1$ or $s \|e e$.
$\conds(m-1-s,m-1-e)$ holds since $x[m-1-s-(m-1-e)+1 \pp m-1-(m-1-e)]=x[f-s+1 \pp e]$ is a suffix of $x$

If $s \le e$ then $x[e-s] \ne x[m-1-s]$ and
 and $\condo((m-1-s,m-1-e)$ holds since $0< m-1-e\le m-1-s$ and $x[m-1-s-(m-1-e)]=x[e-s]\ne =x[m-1-s]$.
 
If $s=e+1$ then
 $\conds(m-1-s,m-1-e)$ holds since $m-1-s < m-1-e$.
 
In both cases $\bsuff[m-1-\suf[e]] \le m-1-e$.

\item
%$\bsuff[m-1-\suf[i]] < m-1-i$ for $f<i<r_3$ and for $\ell_3 \le i < f$

If $e<i<r_3$, $\suf[i]=k_1$ and
$\conds(m-1-k_1,m-1-r_3)$  and $\condo((m-1-k_1,m-1-r_3)$ hold
and $m-1-r_3<m-1-i$.

If $\ell_3 \le i < e$, $\suf[i] = i-\ell_3+1=k$.

$\conds(m-1-k,k_1-k)$ holds since $0<k_1-k\le m-k$ and $x[m-1-k-(k_1-k)+1\pp m-(k_1-k)-1]=x[m-k_1\pp m-k_1+k-1] = a^k$ is a suffix of $x$.
$\condo(m-1-k,k_1-k)$ holds since $0<k_1-k\le m-1-k$ and $x[m-1-k-(k_1-k)] = x[m-k_1-1] = b_1 \ne a = x[m-1-k]$
 and $k_1-k < m-1-i$ thus $\bsuff[m-1-k] < m-1-i$;
\end{enumerate}

\end{proof}

The next lemma shows that positions inside borders that do not correspond to end position of borders
 cannot contribute to new values in the $\bsuff$ table.

\begin{lemma}\label{lemma-border}
Let $x=uvu$ where $u,v\in\Sigma^*$ and $u$ is a border of $x$.
Then for $0\le i < |u|$ such that $\suf[i] \le i$ it holds that
 $\bsuff[m-1-\suf[i]] < m-1-i$.
\end{lemma}

\begin{proof}
Let $k=i+m-|u|$
$\conds(m-1-\suf[i],m-k-1)$ holds since $0<m-k-1< m-\suff[i]-1$ and $x[m-1-\suff[i]-(m-k-1)+1\pp m-(m-k-1)-1]=x[k-\suf[i]+1\pp k]$ is a suffix of $x$.
$\condo(m-1-\suf[i],m-k-1)$ holds since $0<m-k-1\le m-1-k$ and $x[m-1-\suf[i]-(m-k-1)] = x[k-\suf[i]] = b_1 \ne a = x[m-1-\suf[i]]$
 by definition of $\suf$ and because $i$ and $k$ are the same positions in the two occurrences of $u$.
\end{proof}

Then the computation of the table $\bsuff$ can be done by considering runs of $a$'s from right to left.
The rightmost run of length $k_1$ is processed according to Lemma~\ref{lemma-lastrun}.
Subsequent runs of length strictly smaller than are processed according to Lemma~\ref{lemma-smallrun}:
 no position can contribute to a new value of the table $\bsuff$.
Subsequent runs of length greater than are processed according to Lemma~\ref{lemma-largerun}:
 at most two positions can contribute to new values of the table $\bsuff$.
Borders are processed according to Lemma~\ref{lemma-border}.
 
This new method can be done in linear time by considering the two variables $f$ and $g$ that avoid
 comparisons  or in quadratic time by performing all the comparisons for each one position of the runs
 that requires comparisons.

Figures~\ref{figu-algo-ft1goodsuffixes}, \ref{figu-algo-ft2goodsuffixes} and \ref{figu-algo-ft3goodsuffixes} give
 respectively the pseudo-codes for processing the rightmost run, the internal tuns until the longest
 border included and the runs that are located at the left of the longest border.
This corresponds to the linear method.
Its complexity analysis follows the complexity analysis of the \Algo{Suffixes} since it uses the same two 
 variables $f$ and $g$.
Lines~\ref{line-start-small-run} to~\ref{line-end-small-run} in~\figurename{}~\ref{figu-algo-ft2goodsuffixes}
 proceeds runs of $a$'s according to Lemma~\ref{lemma-smallrun} while
 lines~\ref{line-start-large-run} to~\ref{line-end-large-run} in~\figurename{}~\ref{figu-algo-ft2goodsuffixes}
 proceeds runs of $a$'s according to Lemma~\ref{lemma-largerun}.
It uses \Algo{Border} given in \figurename{}~\ref{figu-algo-ftborder} that proceeds according Lemma~\ref{lemma-border}.

Figures~\ref{figu-algo-ft1goodsuffixes} and~\ref{figu-algo-ft2bisgoodsuffixes} present the quadratic method
 that avoids the computation of the table $\suf$.
 
We also give a mixed version in \figurename~\ref{figu-algo-ft2tergoodsuffixes}.
It partially uses the results of the different lemmas of this section.
It processes the rightmost run of $a$'s and the borders as previously but do not analyse
 subsequent runs of $a$'s but is very similar to the classical computation for positions
 between the rightmost run of $a$'s and the longest border of the pattern.
 
%%%%%%%%%%%%%%%%%%%%%%%%%%%%%%%%%%%%%%%%%%%%%%%%%%%%%%%%%%%%%%%%%%%%%%%%%%%%%%%
%%%%%%%%%%%%%%%%%%%%%%%%%%%%%%%%%%%%%%%%%%%%%%%%%%%%%%%%%%%%%%%%%%%%%%%%%%%%%%%
\section{\label{sect-expe}Experimental results}

To evaluate the efficiency of the different methods presented above
we perform several
 experiments with different pattern lengths (from 2 to 1024) and different alphabet sizes
  (2, 4, 20 and 70).
 
%%%%%%%%%%%%%%%%%%%%%%%%%%%%%%%%%%%%%%%%%%%%%%%%%%%%%%%%%%%%%%%%%%%%%%%%%%%%%%%
\subsection{Algorithms}

We have tested five algorithms:
\begin{itemize}
\item
BF: a brute force computation of the $\bsuff$ table;
\item
CL: the classical method (Algorithm \Algo{Good-Suffixes} of \figurename{}~\ref{figu-algo-good-suffixes});
\item
FT1: the linear Fine Tuned method presented in this paper;
\item
FT2: the quadratic Fine Tuned method presented in this paper;
\item
FT1: the mixed Fine Tuned method presented in this paper.
\end{itemize}

These algorithms have been coded in C in an homogeneous way to keep the 
 comparison significant.
All codes have been compiled with \texttt{gcc} with the \texttt{-O3} optimization
 options.
The machine we used has an Intel Core i5 processor at 1.1GHz
 running macOS Big Sur Version 11.7.10.
The codes of the algorithms is available on \url{https://github.com/lecroq/goodsuff}.

%%%%%%%%%%%%%%%%%%%%%%%%%%%%%%%%%%%%%%%%%%%%%%%%%%%%%%%%%%%%%%%%%%%%%%%%%%%%%%%
\subsection{Results}

Tables~\ref{table-2} to~\ref{table-70} shows the running times of the different algorithms
 (fastest times are in bold face).
Figures~\ref{figu-alpha2} to~\ref{figu-alpha2} show graph of the results where the brute force method has been excluded
 since it is too slow.
Running times correspond to 10,000 runs on 1000 different strings pseudo-randomly generated, there given in ms.

These results show that the new methods are the fastest in almost all cases.
For small size alphabet (2 and 4) the quadratic version is the fastest
 while for large size alphabets (20 and 70) the mixed version is the fastest.

%%%%%%%%%%%%%%%%%%%%%%%%%%%%%%%%%%%%%%%%%%%%%%%%%%%%%%%%%%%%%%%%%%%%%%%%%%%%%%%
%%%%%%%%%%%%%%%%%%%%%%%%%%%%%%%%%%%%%%%%%%%%%%%%%%%%%%%%%%%%%%%%%%%%%%%%%%%%%%%
\section{\label{sect-conc}Conclusion}

In this article we presented methods for computing the $\bsuff$ table that is used
 for shifting the pattern in the classical Boyer-Moore exact string matching algorithm.
The new methods are based on the processing of runs of the last letter in the pattern.
However only a mixed method based on the new and the classical method reveals faster
 than the classical one.
It may be that possible tuning could be enable to find another mixed version faster
 than the one presented in this article.
It also gives a new methods for computing only the $\suf$ table.
Being based on analysis of runs, the new method should probably used for computing
 the $\bsuff$ table for a pattern given in a RLE (Run Length Encoding) form.

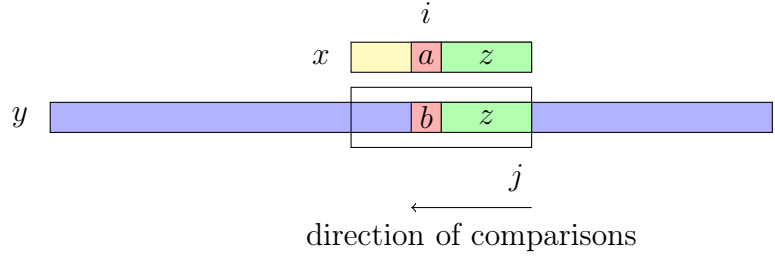
\begin{figure}
\begin{center}
\begin{tikzpicture}[scale=.8]
\drawColoredWord{5}{1}{3}{}{yellow}
\drawColoredWord{6}{1}{.5}{$a$}{red}
\drawColoredWord{6.5}{1}{1.5}{$z$}{green}
\drawColoredWord{0}{0}{12}{}{blue}
\drawColoredWord{6}{0}{.5}{$b$}{red}
\drawColoredWord{6.5}{0}{1.5}{$z$}{green}
\draw (5,-.25) rectangle (8,.75);
\draw (-.5,.25) node {$y$};
\draw (4.5,1.25) node {$x$};

\draw (6.25,2) node {$i$};
\draw (7.75,-.75) node {$j$};
\draw[<-] (6,-1.25) -- (8,-1.25);
\draw (7,-1.75) node {direction of comparisons};
\end{tikzpicture}
\end{center}
\caption{\label{figu-situ-bm}
Suffix $z = x[i+1\pp m-1]$ of $x$ is equal to factor $z = y[j-m+i+2\pp j]$ of $y$
 and symbol $a = x[i]$ is different from symbol $b = y[j-m+i+1]$.
}
\end{figure}

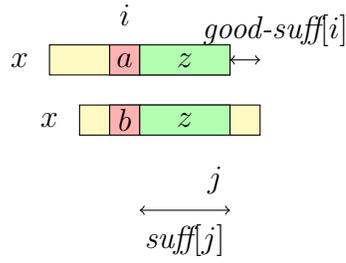
\begin{figure}
\begin{center}
\begin{tikzpicture}[scale=.8]
\drawColoredWord{0}{1}{3}{}{yellow}
\drawColoredWord{1}{1}{.5}{$a$}{red}
\drawColoredWord{1.5}{1}{1.5}{$z$}{green}
\drawColoredWord{0.5}{0}{3}{}{yellow}
\drawColoredWord{1}{0}{.5}{$b$}{red}
\drawColoredWord{1.5}{0}{1.5}{$z$}{green}
\draw (0,.25) node {$x$};
\draw (-.5,1.25) node {$x$};

\draw (1.25,2) node {$i$};
\draw (2.75,-.75) node {$j$};
\draw[<->] (1.5,-1.25) -- (3,-1.25);
\draw (2.25,-1.75) node {$\suf[j]$};

\draw[<->] (3,1.25) -- (3.5,1.25);
\draw (3.75,1.75) node {$\bsuff[i]$};
\end{tikzpicture}
\end{center}
\caption{\label{figu-rel-tables}
$\bsuff[i] = \min\{m-1-j \mid m-1 - \suf[j] = i\}$.
}
\end{figure}

\begin{figure}
\begin{center}
\begin{tikzpicture}[scale=.8]
\drawColoredWord{0}{1}{3}{}{yellow}
\drawColoredWord{1}{1}{.5}{$a$}{red}
\drawColoredWord{1.5}{1}{1.5}{$z$}{green}
\drawColoredWord{2}{0}{3}{}{yellow}
\drawColoredWord{2}{0}{1}{}{green}
\draw (-.5,1.25) node {$x$};
\draw (1.5,.25) node {$x$};

\draw (1.25,2) node {$i$};
\draw (2.75,-.75) node {$j$};
\draw[<->] (2,-1.25) -- (3,-1.25);
\draw (2.5,-1.75) node {$\suf[j]$};

\draw[<->] (3,1.25) -- (5,1.25);
\draw (4,1.75) node {$m-1-j$};
\end{tikzpicture}
\end{center}
\caption{\label{figu-rel-tables2}
$\bsuff[i] = m-j-1$.
}
\end{figure}

\begin{figure}
\begin{algo}{Suffixes}{x,m}
        \SET{g}{m-1}
        \SET{\suf[m-1]}{m}
        \DOFORI{i}{m-2}{0}
                \IF{i > g \mbox{ et } \suf[i+m-1-f] \ne i-g}
                        \SET{\suf[i]}{\min\{\suf[i+m-1-f], i-g\}}
                \ELSE
                        \SET{g}{\min\{i,g\}}
                        \SET{f}{i}
                        \DOWHILE{g \ge 0 \mbox{ et } x[g] = x[g+m+1-f]}
                                \SET{g}{g-1}
                        \OD
                        \SET{\suf[i]}{f-g}
                \FI
        \OD
        \RETURN{\suf}
\end{algo}
\caption{\label{figu-algo-suffixes}
\Call{Suffixes}{x,m} returns the table $\suf$ of string $x$ of length $m$.
}
\end{figure}
 
\begin{figure}
\begin{algo}{Good-Suffixes}{x,m,\suf}
        \SET{i}{0}
        \DOFORI{j}{m-2}{-1}
                \IF{j = -1 \mbox{ ou } \suf[j] = j+1}
                        \DOWHILE{i < m-1-j}
                                \SET{\bsuff[i]}{m-1-j}
                                \SET{i}{i+1}
                        \OD
                \FI
        \OD
        \DOFORI{j}{0}{m-2}
                \SET{\bsuff[m-1-\suf[j]]}{m-1-j}
        \OD
        \RETURN{\bsuff}
\end{algo}
\caption{\label{figu-algo-good-suffixes}
\Call{Good-Suffixes}{x,m,\suf} returns the table $\bsuff$ of string $x$ of length $m$
 given its table $\suf$.
}
\end{figure}

\begin{figure}
\begin{center}
\begin{tikzpicture}[scale=.8]
\drawColoredWord{0}{1}{8}{\small{$x_1$}}{yellow}
\drawColoredWord{8}{1}{.5}{\small{$b_1$}}{red}
\drawColoredWord{8.5}{1}{2}{\small{$a^{k_1}$}}{green}
\draw (-.5,1.25) node {$x$};

\draw (8.75,2.1) node {$\ell_1$};
\draw (10,2) node {$r_1$};
\draw (9.5,2.5) node {$i$};

\draw[<->] (8.5,0.75) -- (10.5,0.75);
\draw (9.5,0.4) node {$k_1$};
\draw[<->] (8.5,0) -- (9.5,0);
\draw (9,-0.25) node {$k$};

\end{tikzpicture}
\end{center}
\caption{\label{figu-lemma1}
$x$ ends with a run of $k_1$ a's.}
\end{figure}
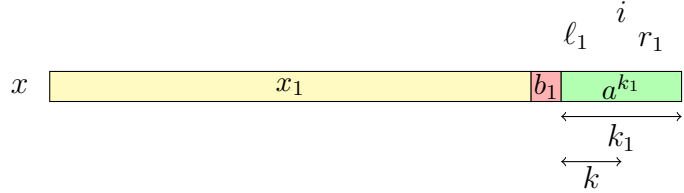

\begin{figure}
\begin{center}
\begin{tikzpicture}[scale=.8]
\drawColoredWord{0}{1}{3.5}{\small{$x_2$}}{yellow}
\drawColoredWord{3.5}{1}{.5}{\small{$b_2$}}{red}
\drawColoredWord{4}{1}{1.5}{\small{$a^{k_2}$}}{green}
\drawColoredWord{5.5}{1}{2.5}{\small{$x_3$}}{yellow}
\drawColoredWord{8}{1}{.5}{\small{$b_1$}}{red}
\drawColoredWord{8.5}{1}{2}{\small{$a^{k_1}$}}{green}
\draw (-.5,1.25) node {$x$};

\draw (4.25,2.1) node {$\ell_2$};
\draw (5.25,2) node {$r_2$};

\draw[<->] (4,0.75) -- (5.5,0.75);
\draw (4.75,0.25) node {$k_2$};

\draw[<->] (8.5,0.75) -- (10.5,0.75);
\draw (9.5,0.25) node {$k_1$};
\end{tikzpicture}
\end{center}
\caption{\label{figu-lemma2}
$x$ contains  a run of $a$'s of length $k_2 < k_1$.}
\end{figure}
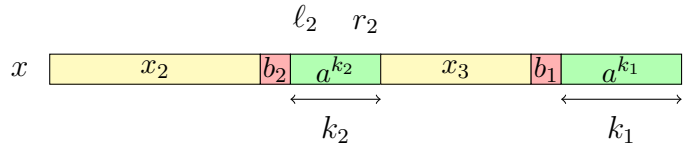

\begin{figure}
\begin{center}
\begin{tikzpicture}[scale=.8]
\drawColoredWord{0}{1}{2.5}{\small{$x_4$}}{yellow}
\drawColoredWord{2.5}{1}{.5}{\small{$b_4$}}{red}
\drawColoredWord{3}{1}{2.5}{\small{$a^{k_3}$}}{green}
\drawColoredWord{5.5}{1}{2.5}{\small{$x_5$}}{yellow}
\drawColoredWord{8}{1}{.5}{\small{$b_1$}}{red}
\drawColoredWord{8.5}{1}{2}{\small{$a^{k_1}$}}{green}
\draw (-.5,1.25) node {$x$};

\draw (3.25,2.1) node {$\ell_3$};
\draw (5.25,2) node {$r_3$};

\draw (4.75,2.8) node {$e$};
\draw[<->] (3,2.5) -- (5,2.5);
\draw[<->] (3,0.75) -- (5.5,0.75);
\draw (4.25,0.25) node {$k_3$};

\draw[<->] (8.5,0.75) -- (10.5,0.75);
\draw (9.5,0.25) node {$k_1$};
\end{tikzpicture}
\end{center}
\caption{\label{figu-lemma3}
$x$ contains  a run of $a$'s of length $k_3 \ge k_1$.}
\end{figure}

\begin{figure}
\begin{center}
\begin{small}
\setlength{\tabcolsep}{1.5pt}
\begin{tabular}{|l|c|c|c|c|c|c|c|c|c|c|c|c|c|c|c|c|c|c|c|c|c|c|c|c|c|c|}
\hline
$i$ & 0 & 1 & 2 & 3 & 4 & 5 & 6 & 7 & 8 & 9 & 10 & 11 & 12 & 13 & 14 & 15 & 16 & 17 & 18 & 19 & 20 & 21 & 22 & 23 & 24 & 25\\
\hline
$x[i]$   &\sa{a}&\sa{a}&\sa{b}&\sa{b}&\sa{a}&\sa{a}&\sa{a}&\sa{a}&\sa{b}&\sa{b}&\sa{a}&\sa{a}&\sa{a}&\sa{a}&\sa{a}&\sa{b}&\sa{b}&\sa{a}&\sa{a}&\sa{a}&\sa{b}&\sa{b}&\sa{a}&\sa{a}&\sa{a}&\sa{a}\\
\hline
$\suf[i]$&1&2&0&0&1&2&3&8&0&0&1&2&3&9&4&0&0&1&2&3&0&0&1&2&3&26\\
\hline
first&&&&&&&&&&&&&&&&&&&&&&&&&&\\
loop&&&&&&&&&&&&&&&&&&&&&&&&&&\\
$j=7$&\textbf{18}&\textbf{18}&\textbf{18}&\textbf{18}&\textbf{18}&\textbf{18}&\textbf{18}&\textbf{18}&\textbf{18}&\textbf{18}&\textbf{18}&\textbf{18}&\textbf{18}&\textbf{18}&\textbf{18}&\textbf{18}&18&18&&&&&&&&\\
$j=1$&&&&&&&&&&&&&&&&&&&\textbf{24}&\textbf{24}&\textbf{24}&24&24&24&&\\
$j=0$&&&&&&&&&&&&&&&&&&&&&&&&&25&\\
\hline
second&&&&&&&&&&&&&&&&&&&&&&&&&&\\
loop&&&&&&&&&&&&&&&&&&&&&&&&&&\\
$j=-1$&&&&&&&&&&&&&&&&&&&&&&&&&&26\\
$j=0$&&&&&&&&&&&&&&&&&&&&&&&&&25&\\
$j=1$&&&&&&&&&&&&&&&&&&&&&&&&24&&\\
$j=2$&&&&&&&&&&&&&&&&&&&&&&&&&&23\\
$j=3$&&&&&&&&&&&&&&&&&&&&&&&&&&22\\
$j=4$&&&&&&&&&&&&&&&&&&&&&&&&&21&\\
$j=5$&&&&&&&&&&&&&&&&&&&&&&&&20&&\\
$j=6$&&&&&&&&&&&&&&&&&&&&&&&19&&&\\
$j=7$&&&&&&&&&&&&&&&&&&\textbf{18}&&&&&&&&\\
$j=8$&&&&&&&&&&&&&&&&&&&&&&&&&&17\\
$j=9$&&&&&&&&&&&&&&&&&&&&&&&&&&16\\
$j=10$&&&&&&&&&&&&&&&&&&&&&&&&&15&\\
$j=11$&&&&&&&&&&&&&&&&&&&&&&&&14&&\\
$j=12$&&&&&&&&&&&&&&&&&&&&&&&13&&&\\
$j=13$&&&&&&&&&&&&&&&&&\textbf{12}&&&&&&&&&\\
$j=14$&&&&&&&&&&&&&&&&&&&&&&\textbf{11}&&&&\\
$j=15$&&&&&&&&&&&&&&&&&&&&&&&&&&10\\
$j=16$&&&&&&&&&&&&&&&&&&&&&&&&&&9\\
$j=17$&&&&&&&&&&&&&&&&&&&&&&&&&8&\\
$j=18$&&&&&&&&&&&&&&&&&&&&&&&&7&&\\
$j=19$&&&&&&&&&&&&&&&&&&&&&&&6&&&\\
$j=20$&&&&&&&&&&&&&&&&&&&&&&&&&&5\\
$j=21$&&&&&&&&&&&&&&&&&&&&&&&&&&\textbf{4}\\
$j=22$&&&&&&&&&&&&&&&&&&&&&&&&&\textbf{3}&\\
$j=23$&&&&&&&&&&&&&&&&&&&&&&&&\textbf{2}&&\\
$j=24$&&&&&&&&&&&&&&&&&&&&&&&\textbf{1}&&&\\
\hline
$\bsuff[i]$&18&18&18&18&18&18&18&18&18&18&18&18&18&18&18&18&12&18&24&24&24&11&1&2&3&4\\
\hline
\end{tabular}
\end{small}
\end{center}
\caption{\label{figu-example1}
Computation of table $\bsuff$ for 
 $x = \sa{aabbaaaabbaaaaabbaaabbaaaa}$
 with the help of table $\suf$.
 The first and second loops refer respectively to the first and the second \textbf{for} loops of Algorithm \Algo{Good-Suffixes}.
Final values are highlighted in bold face.
}
\end{figure}

\begin{figure}
\begin{algo}{FinedTunedGoodSuffix-PART-I}{x,m}
  \SET{i}{m-2}
 % \COM{Looking for the start position of the last run of a's}
  \DOWHILE{i \geq 0 \mbox{ and } x[i] = x[m-1]}
    \SET{i}{i-1}
  \OD
  \COM{According to Lemma~\ref{lemma-lastrun}}
  \SET{(\bsuff[m-1],k_1)}{(m-1-i,m-1-i)}
  \SET{(j,k)}{(m-2,k_1-1)}
  \DOWHILE{j>i}
    \SET{(\suf[j],\bsuff[j])}{(k,k)}
    \SET{(j,k)}{(j-1,k-1)}
  \OD
  \DOFORI{j}{0}{i}
    \SET{\bsuff[j]}{m}
  \OD
\end{algo}
\caption{\label{figu-algo-ft1goodsuffixes}
\Call{FineTunedGoodSuffix-PART-I}{x,m} processes the rightmost run of $a$'s of string $x$ of length $m$.
}
\end{figure}

\begin{figure}
\begin{footnotesize}
\begin{algo}{FineTunedGoodSuffix-PART-II}{x,m}
  \SET{(g,j)}{(m,0)}
  \DOWHILE{\mbox{\texttt{true}}}
   \SET{i}{\mbox{position of the next $a$ or $-1$}}
   \IF{i < 0}
      \RETURN{}
    \FI
    \SET{r}{i}
   %   \COM{Looking for the start position of the current run of a's}
    \SET{i}{\mbox{position of the next letter $\ne a$ or $-1$}}
    \IF{i < 0}
%        \COM{$\suf[r] = r+1$}
       \CALL{Border}{\min\{k_1,r+1\}}
       \IF{r \ge k_1}
          \SET{\bsuff[m-1-k_1]}{\min\{\bsuff[m-1-k_1],m-1-r\}}
        \FI\label{line-end-border1}
        \RETURN{}
    \FI
      %\COM{There is an internal run of a's}
    \SET{k_2}{r-i}
    \IF{k_2 < k_1}
      \SET{(h,k)}{(r,k_2)}\label{line-start-small-run}
      \DOWHILE{h > e}
        \SET{\suf[h]}{k}
        \SET{(h,k)}{(h-1,k-1)}
      \OD
      \ACT{\textbf{continue}}\label{line-end-small-run}
    \FI
    \SET{e}{i+k_1}\label{line-end-large-run}
    \DOFORI{h}{r}{e+1}
      \SET{\suf[h]}{k_1}
    \OD
    \IF{k_2>k_1}
      \SET{\bsuff[m-1-k_1]}{\min\{\bsuff[m-1-k_1],m-1-r\}\}}
    \FI
    \SET{(h,k)}{(e-1,k_1-1)}
    \DOWHILE{h > i}
      \SET{\suf[h]}{k}
      \SET{(h,k)}{(h-1,k-1)}
    \OD
    \IF{g < e \mbox{ and } \suf[e+m-1-f] \ne e - g}
      \IF{\suf[e+m-1-f] < e-g}
        \SET{\suf[e]}{\suf[e+m-1-f]}
      \ELSE
        \SET{\suf[e]}{e-g}
        \SET{\bsuff[m-1-\suf[e]}{\min\{\bsuff[m-1-\suf[e]],m-1-e\}}
      \FI
      \SET{(f,g)}{(e,i)}
      \DOWHILE{g \ge 0 \mbox{ and } x[g]=x[m-1-f+g]}
        \SET{g}{g-1}
      \OD
      \SET{\suff[f]}{f-g}\label{line-start-large-run}
      \IF{g < 0}
%        \COM{$\suf[f]=f+1$}
        \CALL{Border}{f}
        \SET{(i,r)}{(i-1,f)}
        \BREAK
      \ELSE
        \SET{\bsuff[m-1-\suf[f]]}{\min\{\bsuff[m-1-\suf[f]], m-1-f\}}
      \FI
    \FI
  \OD
\end{algo}
\end{footnotesize}
\caption{\label{figu-algo-ft2goodsuffixes}
\Call{FineTunedGoodSuffix-PART-II}{x,m} processes internal riuns of $a$'s of string $x$ of length $m$.
}
\end{figure}

\begin{figure}
\begin{algo}{FineTunedGoodSuffix-PART-III}{x,m}
  \COM{According to Lemma~\ref{lemma-border}}
  \DOWHILE{i \ge 0}
    \IF{x[i] = x[m-1]}
      \IF{i+1 \le \suf[i+m-1-r]}
        \SET{\suf[i]}{i+1}
        \DOWHILE{j < m-1-i}
          \SET{\bsuff[j]}{\min\{\bsuff[j], m-1-i\}}
          \SET{j}{j+1}
        \OD
      \ELSE
        \SET{\suf[i]}{i+m-1-r}
      \FI
    \FI
    \SET{i}{i-1}
  \OD
\end{algo}
\caption{\label{figu-algo-ft3goodsuffixes}
\Call{FineTunedGoodSuffix-PART-III}{x,m} processes remaining borders of string $x$ of length $m$.
}
\end{figure}

\begin{figure}
\begin{algo}{Border}{v}
        \DOWHILE{j \le m-1-v}
          \SET{\bsuff[j]}{\min\{\bsuff[j],m-1-v\}\}}
          \SET{j}{j+1}
        \OD
\end{algo}
\caption{\label{figu-algo-ftborder}
\Call{Boder}{x,m} processes the border of length $v+1$ of string $x$ of length $m$.
}
\end{figure}

\begin{figure}
\begin{footnotesize}
\begin{algo}{FineTuneGoodSuffix-PART-II-bis}{x,m}
  \SET{(g,j)}{(m,0)}
  \DOWHILE{\mbox{\texttt{true}}}
   \SET{i}{\mbox{position of the next $a$ or $-1$}}
   \IF{i < 0}
      \RETURN{}
    \FI
    \SET{r}{i}
   %   \COM{Looking for the start position of the current run of a's}
    \SET{i}{\mbox{position of the next letter $\ne a$ or $-1$}}
    \IF{i < 0}
%        \COM{$\suf[r] = r+1$}
       \CALL{Border}{\min\{k_1,r+1\}}
       \IF{r \ge k_1}
          \SET{\bsuff[m-1-k_1]}{\min\{\bsuff[m-1-k_1],m-1-r\}}
        \FI
        \RETURN{}
    \FI
      %\COM{There is an internal run of a's}
    \SET{k_2}{r-i}
    \IF{k_2 < k_1}
      \ACT{\textbf{continue}}
    \FI
    \SET{e}{i+k_1}
    \IF{k_2>k_1}
      \SET{\bsuff[m-1-k_1]}{\min\{\bsuff[m-1-k_1],m-1-r\}\}}
    \FI
   \SET{(f,g)}{(e,i)}
   \DOWHILE{g \ge 0 \mbox{ and } x[g]=x[m-1-f+g]}
      \SET{g}{g-1}
    \OD
    \IF{g < 0}
%        \COM{$\suf[f]=f+1$}
      \CALL{Border}{f}
    \ELSE
      \SET{\bsuff[m-1-f+g]}{\min\{\bsuff[m-1-f+g], m-1-f\}}
    \FI
  \OD
\end{algo}
\end{footnotesize}
\caption{\label{figu-algo-ft2bisgoodsuffixes}
\Call{FineTuneGoodSuffix-PART-II-bis}{x,m} returns the table $\bsuff$ of string $x$ of length $m$
 without computing the table $\suf$.
}
\end{figure}

\begin{figure}
\begin{footnotesize}
\begin{algo}{FineTuneGoodSuffix-PART-II-ter}{x,m}
 \SET{i}{\mbox{position of the next $a$ or $-1$}}
 \IF{i < 0}
    \RETURN{}
  \FI
  \SET{(f,g)}{(i,i-1)}
  \DOWHILE{g \ge 0 \mbox{ and } x[g]=x[m-1-f+g]}
    \SET{g}{g-1}
  \OD
  \SET{\suff[f]}{f-g}
  \IF{g < 0}
%        \COM{$\suf[f]=f+1$}
    \CALL{Border}{f}
    \SET{i}{i-1}
    \ACT{\textbf{goto } lastStep}
  \ELSE
    \SET{\bsuff[m-1-\suf[i]]}{\min\{\bsuff[m-1-\suf[i]], m-1-i\}}
  \FI

  \DECR{i}
  \DOWHILE{i \ge 0}
    \DECR{i}
    \IF{x[i] = x[m-1]}
      \IF{g < i \mbox{ and } \suf[e+m-1-i] \ne i - g}
        \IF{\suf[e+m-1-i] < i-g}
          \SET{\suf[i]}{\suf[i+m-1-f]}
        \ELSE
          \SET{\suf[i]}{i-g}
          \SET{\bsuff[m-1-\suf[i]}{\min\{\bsuff[m-1-\suf[i]],m-1-i\}}
        \FI
      \ELSE
        \SET{(f,g)}{(i,\min\{g,i-1\})}
        \DOWHILE{g \ge 0 \mbox{ and } x[g]=x[m-1-f+g]}
          \SET{g}{g-1}
        \OD
        \SET{\suff[i]}{f-g}
        \IF{g < 0}
%        \COM{$\suf[f]=f+1$}
          \CALL{Border}{i}
          \DECR{i}
          \BREAK
        \ELSE
          \SET{\bsuff[m-1-\suf[i]]}{\min\{\bsuff[m-1-\suf[i]], m-1-i\}}
        \FI
      \FI
    \FI
    \DECR{i}
  \OD
  \ACT{lastStep:}
  \DOWHILE{i \ge 0}
    \IF{x[i] = x[m-1]}
      \IF{i+1 \le \suf[i+m-1-f]}
        \SET{\suf[i]}{i+1}
        \CALL{Border}{i}
      \ELSE
        \SET{\suf[i]}{i+m-1-f}
      \FI
    \FI
    \SET{i}{i-1}
  \OD
\end{algo}
\end{footnotesize}
\caption{\label{figu-algo-ft2tergoodsuffixes}
\Call{FineTuneGoodSuffix-PART-II-ter}{x,m} computes the table $\bsuff$ of string $x$ of length $m$
 using a mixed strategy.
}
\end{figure}

\begin{table}
\caption{\label{table-2} 
Experimental results for an alphabet of size
$\sigma = 2$.
}
\begin{center}
\begin{tabular}{|l|c|c|c|c|c|}
\hline
$m$ & BF & CL & FT1 & FT2 & FT3\\
\hline
2 & 0.008745 & 0.008731 & 0.008772 & 0.008589 & \textbf{0.008573}\\
4 & 0.005412 & 0.005586 & 0.005482 & \textbf{0.005218} & 0.005512\\
8 & 0.007156 & 0.006581 & 0.006152 & \textbf{0.005743} & 0.005952\\
16 & 0.013612 & 0.006990 & 0.007355 & \textbf{0.006785} & 0.007183\\
32 & 0.033180 & 0.008698 & 0.009581 & \textbf{0.008199} & 0.009244\\
64 & 0.111221 & 0.011614 & 0.013312 & \textbf{0.010890} & 0.012055\\
128 & 0.434411 & 0.017598 & 0.020186 & \textbf{0.015496} & 0.017973\\
256 & 1.712719 & 0.028663 & 0.032937 & \textbf{0.024166} & 0.029018\\
512 & 6.929127 & 0.050306 & 0.057030 & \textbf{0.041220} & 0.049889\\
1024 & 27.308697 & 0.090039 & 0.099594 & \textbf{0.071921} & 0.088235\\
\hline
\end{tabular}
\end{center}
\end{table}

\begin{table}
\caption{\label{table-4} 
Experimental results for an alphabet of size
$\sigma = 4$.
}
\begin{center}
\begin{tabular}{|l|c|c|c|c|c|c|}
\hline
$m$ & BF & CL & FT1 & FT2 & FT3\\
\hline
2 & 0.005104 & 0.005034 & 0.005107 & \textbf{0.004976} & 0.004999\\
4 & 0.005784 & 0.005678 & 0.005702 & 0.005482 & \textbf{0.005321}\\
8 & 0.006879 & 0.005888 & 0.005803 & \textbf{0.005748} & 0.005764\\
16 & 0.012003 & 0.006450 & 0.006716 & \textbf{0.006296} & 0.006322\\
32 & 0.028193 & 0.008150 & 0.008162 & \textbf{0.007080} & 0.007388\\
64 & 0.088753 & 0.009949 & 0.010225 & \textbf{0.008603} & 0.008775\\
128 & 0.324062 & 0.014428 & 0.014754 & \textbf{0.012020} & 0.012106\\
256 & 1.257153 & 0.023508 & 0.023479 & \textbf{0.018005} & 0.018597\\
512 & 5.032720 & 0.041383 & 0.040230 & \textbf{0.029101} & 0.030979\\
1024 & 20.586937 & 0.075417 & 0.070922 & \textbf{0.051622} & 0.056580\\
\hline
\end{tabular}
\end{center}
\end{table}

\begin{table}
\caption{\label{table-20}
Experimental results for an alphabet of size
$\sigma = 20$.
}
\begin{center}
\begin{tabular}{|l|c|c|c|c|c|}
\hline
$m$ & BF & CL & FT1 & FT2 & FT3\\
\hline
2 & 0.005233 & 0.005077 & 0.005256 & \textbf{0.005024} & 0.005093\\
4 & 0.005434 & 0.005361 & 0.005361 & \textbf{0.005259} & 0.005276\\
8 & 0.006379 & 0.005757 & 0.005432 & \textbf{0.005250} & 0.005259\\
16 & 0.009964 & 0.005923 & 0.005603 & \textbf{0.005406} & 0.005436\\
32 & 0.025321 & 0.006496 & 0.006112 & 0.006406 & \textbf{0.005582}\\
64 & 0.084819 & 0.007816 & 0.007029 & 0.006477 & \textbf{0.006461}\\
128 & 0.316767 & 0.010470 & 0.008664 & \textbf{0.007432} & 0.007606\\
256 & 1.230335 & 0.015584 & 0.012013 & 0.010109 & \textbf{0.009604}\\
512 & 4.902985 & 0.025234 & 0.018403 & 0.014375 & \textbf{0.013772}\\
1024 & 20.015751 & 0.044632 & 0.030326 & 0.023298 & \textbf{0.022046}\\
\hline
\end{tabular}
\end{center}
\end{table}

\begin{table}
\caption{\label{table-70}
Experimental results for an alphabet of size
$\sigma = 70$.
}
\begin{center}
\begin{tabular}{|l|c|c|c|c|c|}
\hline
$m$ & BF & CL & FT1 & FT2 & FT3\\
\hline
2 & 0.005144 & 0.005125 & 0.005142 & 0.005057 & \textbf{0.004933}\\
4 & 0.005428 & 0.005364 & 0.005449 & 0.005239 & \textbf{0.005205}\\
8 & 0.006245 & 0.005771 & 0.005371 & 0.005772 & \textbf{0.005400}\\
16 & 0.009442 & 0.005767 & 0.005626 & \textbf{0.005400} & 0.005407\\
32 & 0.023653 & 0.006466 & 0.005857 & \textbf{0.005697} & 0.005860\\
64 & 0.078086 & 0.007575 & 0.006712 & 0.006249 & \textbf{0.005986}\\
128 & 0.289294 & 0.009746 & 0.007809 & 0.007066 & \textbf{0.006884}\\
256 & 1.136308 & 0.014193 & 0.010026 & 0.008646 & \textbf{0.008529}\\
512 & 4.533947 & 0.022483 & 0.014530 & 0.012153 & \textbf{0.011148}\\
1024 & 18.628668 & 0.041737 & 0.024159 & 0.019512 & \textbf{0.017977}\\
\hline
\end{tabular}
\end{center}
\end{table}

\begin{figure}
\begin{center}
\includegraphics[angle=-90,width=12cm]{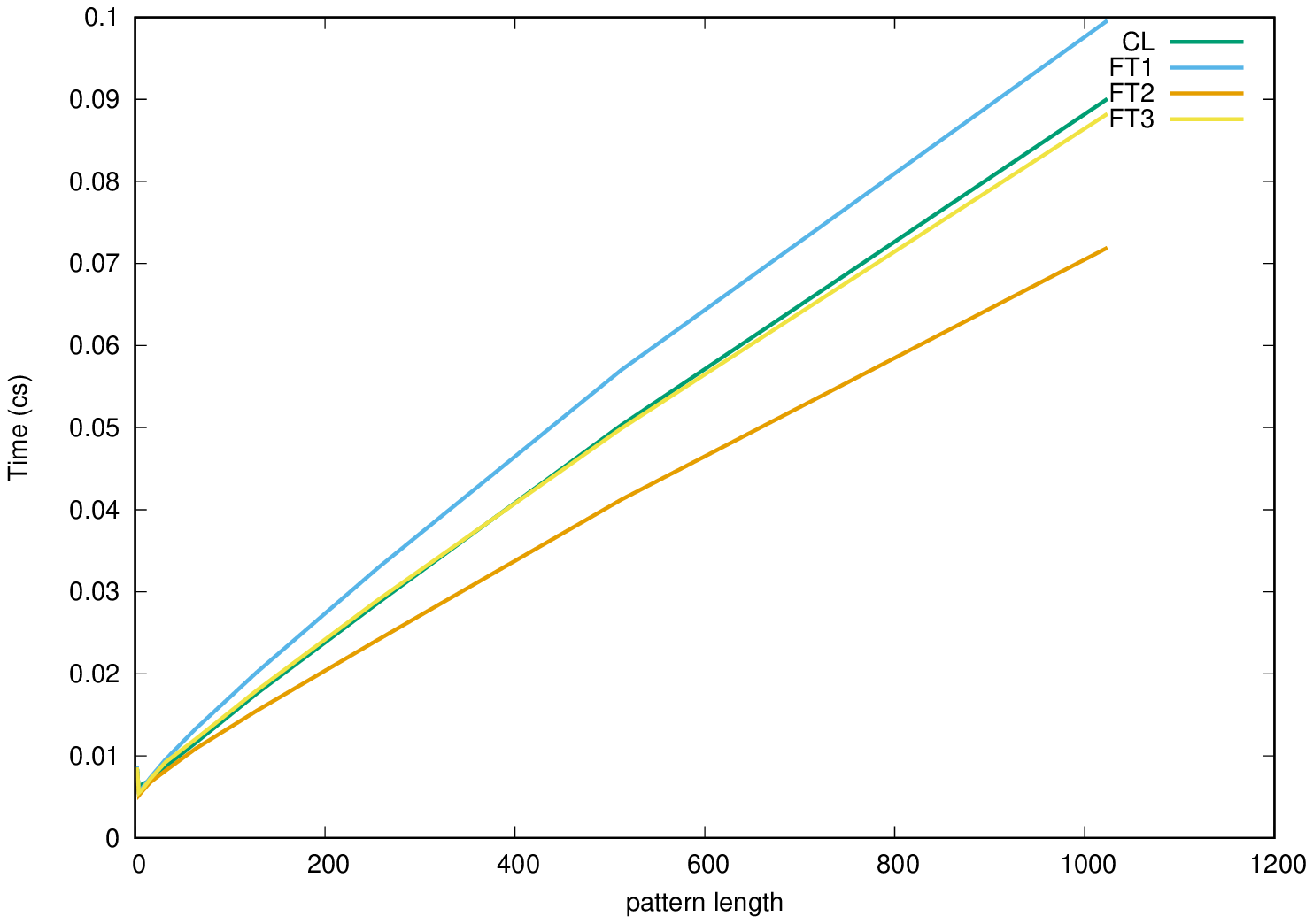}
\end{center}
\caption{\label{figu-alpha2}Experimental results an alphabet of size $\sigma=2$.}
\end{figure}

\begin{figure}
\begin{center}
\includegraphics[angle=-90,width=12cm]{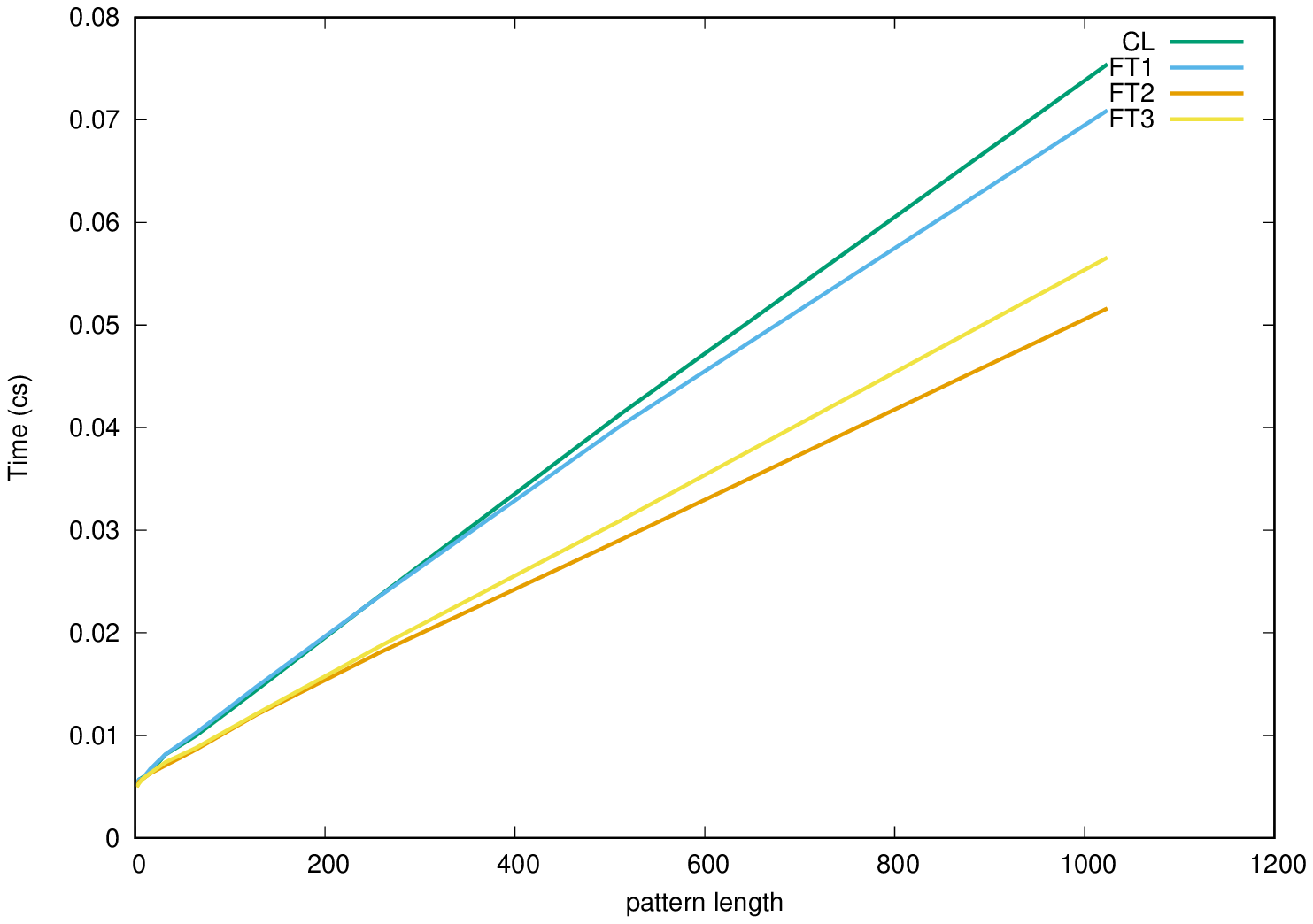}
\end{center}
\caption{\label{figu-alpha4}Experimental results an alphabet of size $\sigma=4$.}
\end{figure}

\begin{figure}
\begin{center}
\includegraphics[angle=-90,width=12cm]{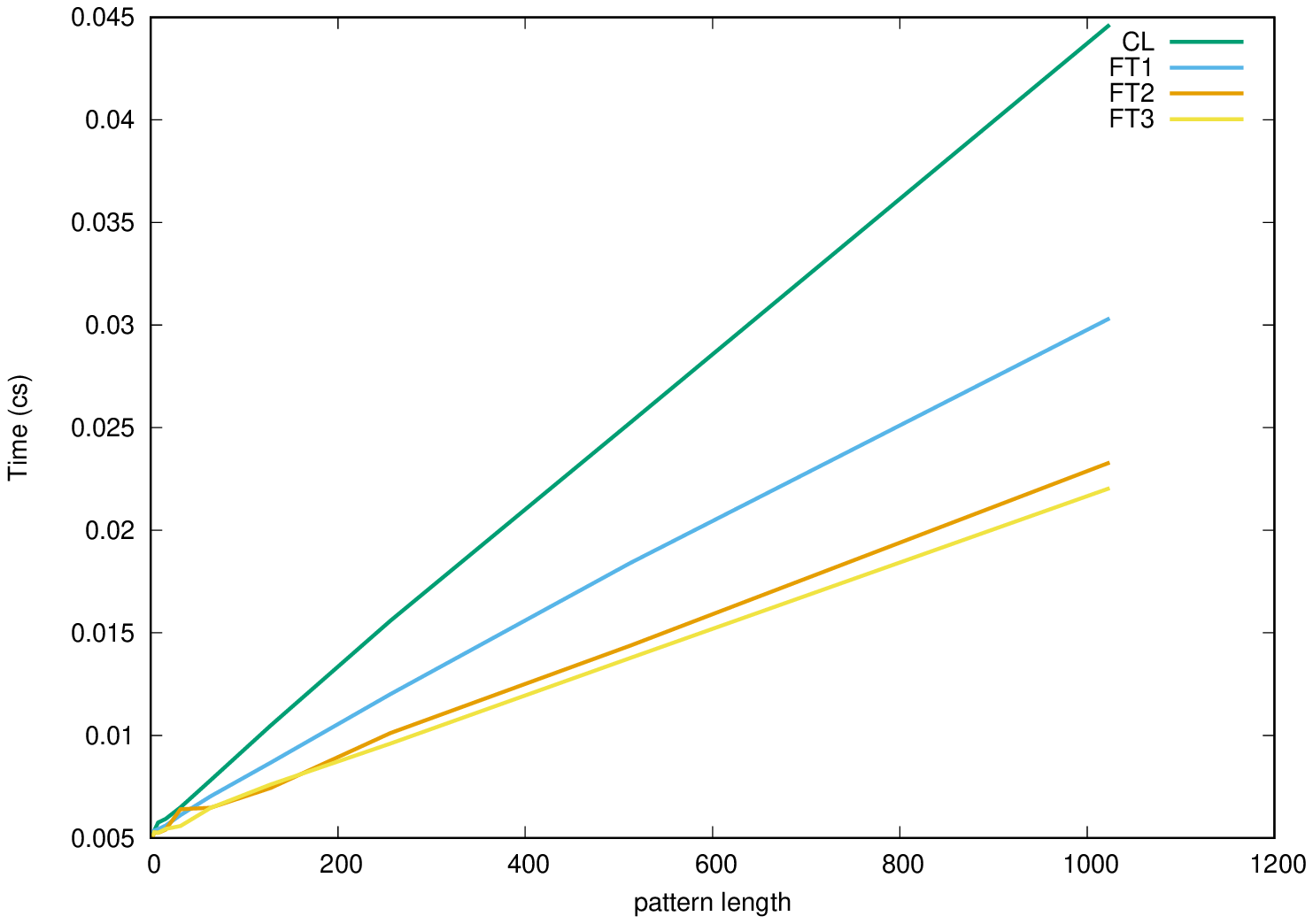}
\end{center}
\caption{\label{figu-alpha20}Experimental results an alphabet of size $\sigma=20$.}
\end{figure}

\begin{figure}
\begin{center}
\includegraphics[angle=-90,width=12cm]{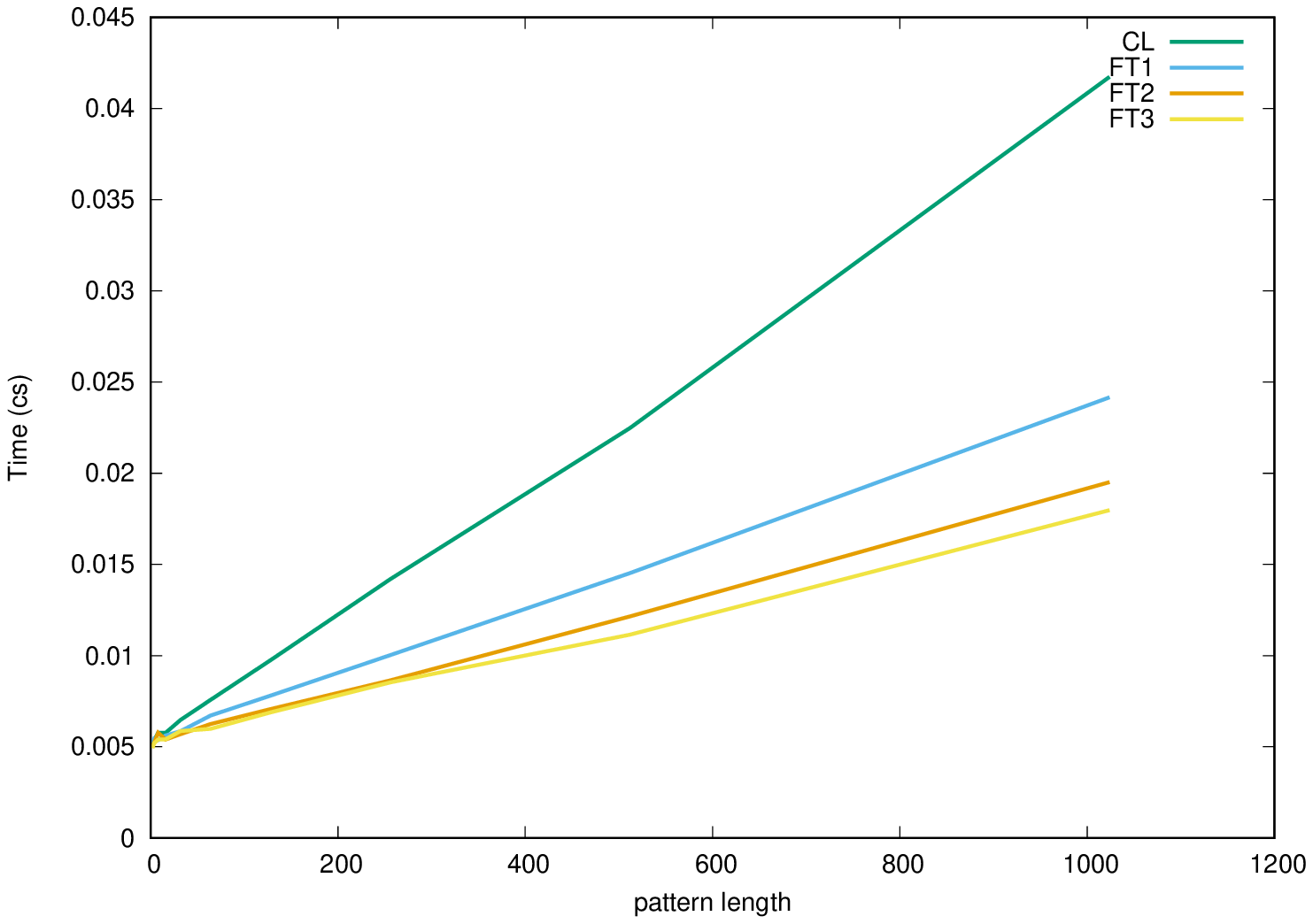}
\end{center}
\caption{\label{figu-alpha70}Experimental results an alphabet of size $\sigma=70$.}
\end{figure}

\bibliographystyle{abbrv}
\bibliography{main}

\end{document}